\documentclass[journal]{IEEEtran}

\usepackage{bm}
\usepackage{makeidx}
\usepackage{enumerate}
\usepackage{color}
\usepackage{cite}
\usepackage{amsmath,amsthm}   
\usepackage{amssymb}
\usepackage{multirow}
\usepackage{graphicx}
\usepackage{epstopdf}
\usepackage{balance}
\usepackage{multicol}
\DeclareGraphicsExtensions{.pdf,.jpeg,.png,.jpg,.emf,.eps}
\usepackage{calc}
\usepackage{url}

\usepackage{upref}
\usepackage{comment}
\usepackage{times}
\usepackage{dsfont}
\usepackage{rawfonts}
\usepackage[T1]{fontenc}
\usepackage{latexsym}
\usepackage{amsfonts}

\usepackage{geometry}
\usepackage[font=small,labelfont=bf]{caption}
\usepackage[font=small,labelfont=bf]{subcaption}
\usepackage{optidef}
\usepackage[compress]{cleveref}

\usepackage{xcolor}
\usepackage{tikz,pgfplots}
\usetikzlibrary{calc}
\usetikzlibrary{intersections}
\usetikzlibrary{arrows,shapes}
\usetikzlibrary{positioning}

\newtheorem{theorem}{Theorem}
\newtheorem{corollary}{Corollary}
\newtheorem{proposition}{Proposition}

\newtheorem{remark}{Remark}

\theoremstyle{definition}

\def\M{{\bf M}}
\def\K{{\bf K}}
\def\Q{{\bf Q}}
\def\I{{\bf I}}
\def\A{{\bf A}}
\def\U{{\bf U}}
\def\V{{\bf V}}
\def\F{{\bf F}}
\def\T{{\bf T}}
\def\G{{\bf G}}
\def\H{{\bf H}}
\def\U{{\bf U}}

\def\g{{\bf g}}
\def\u{{\bf u}}
\def\v{{\bf v}}
\def\h{{\bf h}}
\def\t{{\bf t}}
\def\m{{\bf m}}
\def\Thetab{\bm{\Theta}}
\def\Sigmab{\bm{\Sigma}}
\def\Lambdab{\bm{\Lambda}}

\def\tr{\operatorname{tr}}

\def\diag{\operatorname{diag}}

\begin{document}

\title{SNR Maximization in Beyond Diagonal RIS-assisted Single and Multiple Antenna Links}

\author{Ignacio Santamaria, \IEEEmembership{Senior Member, IEEE}, Mohammad Soleymani, \IEEEmembership{Member, IEEE},
Eduard Jorswieck, \IEEEmembership{Fellow, IEEE},
Jes{\'u}s Guti{\'e}rrez, \IEEEmembership{Member, IEEE}
\thanks{I.~Santamaria is with the Department of Communications Engineering, Universidad de Cantabria, Santander, Spain (e-mail: i.santamaria@unican.es). }
\thanks{M. Soleymani is with the Signal and System Theory Group, Universit{\"a}t  Paderborn, 33098 Paderborn, Germany (e-mail : mohammad.soleymani@sst.upb.de).}

\thanks{E. Jorswieck is with the Institute for Communications Technology, Technische Universität Braunschweig, 38106 Braunschweig,
Germany (e-mail: e.jorswieck@tu-bs.de).}
\thanks{J. Guti{\'e}rrez is with IHP - Leibniz-Institut
f{\"u}r Innovative Mikroelektronik, 15236 Frankfurt (Oder), Germany (e-mail: teran@ihp-microelectronics.com).}

\thanks{The work of I. Santamaria was supported by Ministerio de Ciencia e Innovaci{\'{o}}n and AEI (10.13039/501100011033), under grant PID2019-104958RB-C43 (ADELE). The work of Eduard Jorswieck was supported by
the Federal Ministry of Education and Research (BMBF, Germany) through
the Program of “Souverän. Digital. Vernetzt.” joint Project 6G-RIC, under
Grants 16KISK020K and 16KISK031. }}

\markboth{IEEE Signal Processing Letters}{}

\maketitle

\begin{abstract}
Reconfigurable intelligent surface (RIS) architectures not limited to diagonal phase shift matrices have recently been considered to increase their flexibility in shaping the wireless channel. One of these beyond-diagonal RIS or BD-RIS architectures leads to a unitary and symmetric RIS matrix. In this letter, we consider the problem of maximizing the signal-to-noise ratio (SNR) in single and multiple antenna links assisted by a BD-RIS. The Max-SNR problem admits a closed-form solution based on the Takagi factorization of a certain complex and symmetric matrix. This allows us to solve the max-SNR problem for SISO, SIMO, and MISO channels.

\end{abstract}


\IEEEpeerreviewmaketitle

\section{Introduction}
Reconfigurable intelligent surfaces (RISs) are recently receiving a great deal of attention as an enabling technology to increase spectral and energy efficiency in future wireless communication networks \cite{{RenzoJSAC2020},{pan2020multicell},{ZapponeTWT2019},{ZhangTWT2019},{SoleymaniTVT22},{SchoberPoorTCOM2023}}. The conventional structure of a RIS is a surface composed of passive reconfigurable elements where each element can introduce a phase shift. Therefore, the RIS is usually modeled as a diagonal matrix $\bm{\Theta}=\text{diag}\left(e^{j \theta_1}, \cdots,e^{j \theta_M}\right)$. The simple structure of conventional RIS limits its flexibility in modulating the equivalent wireless channel. For this reason, a new architecture called beyond-diagonal RIS (BD-RIS) has recently been proposed in \cite{{ClerckxTWC22a},{ClerckxTWC23}}. In the BD-RIS architecture, $\bm{\Theta}$ is a full matrix that, according to network theory, must satisfy the constraints: $\bm{\Theta} = \bm{\Theta}^T$ (reciprocal)\footnote{The constraint $\bm{\Theta} = \bm{\Theta}^T$ assumes a reciprocal passive network for which the power losses are the same between any pair of ports regardless of the direction of propagation. Passive RISs that do not satisfy the reciprocity property and, therefore, do not lead to symmetric matrices, have also been proposed in \cite{LiTVT22},\cite{ClerckxTWC23}.} and $\bm{\Theta}^H\bm{\Theta} \preceq \I$ (passive). When the RIS impedance network is purely reactive, the RIS matrix must be symmetric $\bm{\Theta} = \bm{\Theta}^T$, and unitary $\bm{\Theta}^H\bm{\Theta} = \I$ (passive lossless).

In this work, we first consider a SISO wireless link assisted by a BD-RIS with $M$ elements. We assume that the direct link between the transmitter and the receiver is blocked so that there is a single Tx-RIS-Rx link. This assumption represents a real scenario in which the direct link is significantly weaker than the RIS-aided link \cite{huang2019reconfigurable, AlouiniMaxMinTWC2020}. The equivalent channel is
\begin{equation*}
h_{eq} =  \h_R^H \Thetab \h_T,
\end{equation*}
where $\h_T \in \mathbb{C}^{M \times 1}$ is the channel from the transmitter to the RIS, $\h_R \in \mathbb{C}^{M \times 1}$ is the channel from the RIS to the receiver and $\bm{\Theta}$ is the $M \times M $ RIS matrix. We aim at finding the passive lossless BD-RIS matrix $\bm{\Theta}$ that maximizes the signal-to-noise-ratio (SNR) at the receiver or, equivalently, the received signal power   
\begin{subequations}
\begin{align}
\label{eq:MaxSNRproblemBDRIS}
({\cal P}_1): \,\max_{\bm{\Theta}}\,\,& |\h_R^H \Thetab \h_T |^2 \\
\text{s.t.}\,\,& \bm{\Theta}^H\bm{\Theta} = \I,   \label{eq:unitaryconst}\\
\,\,& \bm{\Theta} = \bm{\Theta}^T, \label{eq:symmetryconst}
\end{align}
\end{subequations}
where the constraint that makes the problem non-trivial is the symmetry constraint $\bm{\Theta} = \bm{\Theta}^T$.
Problem ${\cal{P}}_1$ is posed in \cite{ClerckxTWC22a}, where the resulting optimization problem is considered to be difficult to solve, hence the authors use a quasi-Newton method to find a solution \cite[pp. 1235]{ClerckxTWC23}. 

In this letter, we show that ${\cal P}_1$ has a closed-form solution in the SISO case based on Takagi's factorization \cite{{Takagi},{HornBook}} -- a special case of the singular value decomposition (SVD) for symmetric matrices--  of a certain complex symmetric matrix. The solution attains the maximum signal power attained by an unconstrained, unitary but not symmetrical, BD-RIS matrix. In addition, we extend our result to systems where either the transmitter or the receiver has multiple antennas. In the multiple-input single-output case (MISO), the scenario is equivalent to a SISO multiple-access channel (MAC) for which the Max-SNR solution also maximizes the sum rate. We conclude the letter with a collection of open problems.

\section{Max-SNR BD-RIS: The SISO case}
Let us begin by reminding the reader of some basic results.
We write the full SVD of the outer product channel matrix as $\h_R \h_T^H = \U \bm{\Lambda} \V^H$, where $\bm{\Lambda} = \diag(\lambda_1,0,\ldots,0)$. Then, it is known that the rank-1 unitary matrix $\u \v^H$, or the full-rank unitary matrix $\U \V^H$, maximize the output power achieving $P_{max}= \lambda_1^2$. 
Clearly, the difficulty of ${\cal P}_1$ stems from the symmetry constraint $\bm{\Theta} = \bm{\Theta}^T$. Its solution is based on Takagi's decomposition of a certain symmetric matrix, so we first review that result.
\begin{theorem}
Let $\A = \A^T$ be an $n \times n$ complex symmetric matrix. Then, there exist an $n \times n$ unitary matrix $\Q$ and an $n\times n$ positive semidefinite diagonal matrix $\bm{ \Sigma} = \diag(\sigma_1,\ldots,\sigma_n)$ such that $\A = \Q \bm{ \Sigma} \Q^T$.
\end{theorem}

This factorization of a complex symmetric matrix is called the Autonne-Takagi factorization, or the Takagi factorization in short, originally proposed by Autonne \cite{Autonne} and Takagi \cite{Takagi} (see also \cite[Chapter 4, Corollary 4.4.4]{HornBook}). The columns of $\Q$ are called the Takagi vectors of $\A$ and the diagonal elements of $\bm{ \Sigma}$ are its Takagi values. Notice that the Takagi values of $\A$ coincide with the singular values of $\A$. We have the following result.

\begin{proposition} \label{lemmaTakagi} Let $\h_R \h_T^H = \lambda_1 \u_R \u_T^H$ be the rank-one outer product matrix between  $\h_R = \|\h_R\| \u_R$ and $\h_T = \|\h_T\| \u_T$, where $\lambda_1 = \|\h_R\|  \|\h_T\|$. Form the rank-2 symmetric complex matrix $\A = \u_R \u_T^H + (\u_R \u_T^H)^T$, and compute its Takagi's factorization as $\A = \Q \bm{ \Sigma} \Q^T$, where $\Q$ is unitary and $\bm{ \Sigma} = \diag(\sigma_1, \sigma_2, 0,\ldots,0)$ are the Takagi values (singular values) of $\A$. Then, the solution of ${\cal P}_1$ is 
\begin{equation}
\label{eq:PhiBDRIS}
  \Thetab = \Q \Q^T.
\end{equation}
Furthermore, the maximum signal power achieved by this solution is $P_{max} = \lambda_1^2$.
\end{proposition}

\begin{proof}
The fact that $\Thetab = \Q \Q^T$ is unitary and symmetric is trivially checked. Let us define the vectors $\g_R = \Q^H \u_R$ and $\g_T = \Q^T \u_T$. Since $\Q$ is unitary $\|{\g}_R\| = 1$ and $\|{\g}_T\| = 1$. With these definitions the equivalent channel is $ \h_R^H \Thetab \h_T =  \h_R^H \Q \Q^T \h_T = \lambda_1 {\g}_R^H {\g}_T$. We want to prove that $ \g_R^H \g_T = 1$, which in turn implies that $\g_T= \g_R$ and that the received signal power is $P_{max} = \lambda_1^2$. To prove that ${\g}_R^H \g_T = 1$ or, equivalently, that 
${\g}_T^H \g_R = 1$, notice that
\begin{equation}
\Q^H \A \Q^* = \g_R \g_T^H + (\g_R \g_T^H)^T = \bm{ \Sigma},
\label{eq:beforethetrace}
\end{equation}
where $\tr( \bm{ \Sigma}) = \tr(\A) = 2$. Taking traces in \eqref{eq:beforethetrace} and applying the circular property of the trace \cite[pp. 360]{Coherence} and $\tr(\g_R \g_T^H) = \tr((\g_R \g_T^H)^T )$, we finally get $2\tr(\g_R \g_T^H) = 2 \tr(\g_T^H \g_R) = 2 \g_T^H \g_R = 2$ and, therefore, $\g_R^H \g_T  =1$ thus proving the result\footnote{The authors thank Alessio Zappone for interesting discussions on the proof.}.
\end{proof}

Although we have considered for simplicity of exposition the case in which the direct channel is blocked, the solution can be extended to the case in which the equivalent channel is $h_{eq} = h + \h_R^H \Thetab \h_T$. In this case, after computing Takagi's factorization of $\A =  \u_R \u_T^H + (\u_R \u_T^H)^T = \Q \bm{ \Sigma} \Q^T$, the solution for the BD-RIS that maximizes the SNR is $ \Thetab = e^{j \angle h} \Q \Q^T$, where $\angle h$ is the phase of the direct link. The resulting channel gain is $|h_{eq}| = |h| + |\h_R^H \Q \Q^T \h_T| = |h| +\lambda_1$.

\begin{remark} The Max-SNR BD-RIS solution is not unique for the SISO case unless the number of BD-RIS elements is $M=2$. To show this, note that $\A = \u_R \u_T^H + (\u_R \u_T^H)^T = \Q \bm{ \Sigma} \Q^T$ is a rank-2 matrix. Now, partition the unitary matrix $\Q$ into signal and noise subspaces as $\Q = [\Q_{signal} | \Q_{noise}]$, where $\Q_{signal}$ contains the first $2$ columns of $\Q$, and $\Q_{noise}$ contains the remaining $M-2$ columns. Then, we may generate a new basis for the noise subspace as $\Q_{noise}' = \Q_{noise} \T$, where $\T$ is an $(M-2)\times (M-2)$ unitary matrix. The new matrix $\Q' = [\Q_{signal} | \Q_{noise}']$ defines a new unitary and symmetric BD-RIS matrix $\Thetab' = \Q' \Q'^T$ that provides the maximum output power $P_{max} = \lambda_1^2$. When $M=2$ there is no noise subspace so the solution is unique up to a complex scaling of the form $e^{j \theta}$.
\end{remark}

\begin{remark}
    Takagi's factorization algorithms \cite{{Qiao08},{Qiao09}} exploit the fact that the matrix is symmetric to improve computational efficiency. However, it is possible to obtain the factorization through a standard SVD of $\A = \u_R \u_T^H + (\u_R \u_T^H)^T = \F \K \G^H$, and then apply the following steps: i) compute $\t = \diag(\F^H \G^*)$; ii) compute $\bm{\phi} = \frac{\angle \t}{2}$, where $\angle \t$ extracts componentwise the angles of the entries of $\t$; iii) calculate $\F' = \F \diag(e^{j\bm{\phi}})$\footnote{With some abuse of notation, $\diag(\M)$ when $\M$ is a matrix denotes a vector with the diagonal components of $\M$, whereas $\diag(\m)$ when $\m$ is a vector denotes a diagonal matrix. We believe there is no confusion possible.}. With these steps, Takagi's factorization is $\A = \F' \K \F'^T$, and the BD-RIS matrix is $\Thetab = \F' \F'^T$.
\end{remark}

\begin{remark} While reviewing this work, the solution proposed for the same scenario in \cite{NeriniTWC2023} came to our attention. The work in \cite{NeriniTWC2023} is, however, based on a different factorization of the BD-RIS matrix: $ \Thetab = \V {\bf D} \V^T$, where $\V$ is a real orthogonal matrix and ${\bf D} = \diag \left(e^{\phi_1}, \ldots, e^{\phi_1} \right)$. It is interesting to note that the factorizations $ \Thetab = \V {\bf D} \V^T$, with $\V$ real and orthogonal and ${\bf D}$ complex diagonal with unit modulus elements, and $ \Thetab = \Q \Q^T$, with $\Q$ unitary complex, are equivalent. Using the factorization $ \Thetab = \V {\bf D} \V^T$ the authors in \cite{NeriniTWC2023} derive closed-form expressions for the vectors of the matrix $\V$ considering separately the cases $M=2$, $M=3$, and $M \geq 4$. In contrast, our solution based on Takagi's factorization can be obtained through an SVD for any $M$. We believe that Takagi factorization is a natural fit to the problem which provides a simple and easily accessible solution.
\end{remark}

\subsection{Group-connected BD-RIS} 
In ${\cal{P}}_1$, it is considered that all ports of the RIS elements are connected to each other, therefore defining a fully-connected BD-RIS architecture, which may complicate the required circuitry. A good trade-off between complexity and performance is provided by the group-connected BD-RIS, in which the $\Thetab$ matrix is block-diagonal with unitary and symmetric blocks \cite{ClerckxTWC22a}, \cite{ClerckxTWC23}. If we divide the $M$ BD-RIS elements into $G$ groups of $M_G = M/G$ elements each, the resulting Max-SNR optimization problem is 
\begin{align}
\label{eq:MaxSNRproblemBDRISGroup}
({\cal P}_2): \,\max_{\bm{\Theta}_1,\ldots,\bm{\Theta}_G}\,\,& \left | \sum_{g=1}^G \h_{R,g}^H \Thetab_g \h_{T,g} \right |^2 \\
\text{s.t.}\,\,& \bm{\Theta}_g^H\bm{\Theta}_g = \I_G, \, \forall g \nonumber \\
\,\,& \bm{\Theta}_g = \bm{\Theta}_g^T, \, \forall g. \nonumber
\end{align}
Performing Takagi's decomposition of the $M_G \times M_G$ matrices $ \u_{R,g} \u_{T,g}^H + (\u_{R,g} \u_{T,g}^H)^T =\Q_g \bm{ \Sigma}_g \Q_g^T $, we get $\bm{\Theta}_g = \Q_g\Q_g^T $ as solution that maximizes $\left|  \h_{R,g}^H \Thetab_g \h_{T,g} \right |^2$. From the proof of Proposition \ref{lemmaTakagi}, we know that  $\h_{R,g}^H \Q_g\Q_g^T \h_{T,g}$ is a positive real value. This means that the $G$ terms within the summation of \eqref{eq:MaxSNRproblemBDRISGroup} add up coherently and, therefore, they are the optimal solution of ${\cal P}_2$. In summary, problem \eqref{eq:MaxSNRproblemBDRISGroup} decouples into $G$ independent subproblems, each of which can be solved by performing Takagi's decomposition of the $M_G \times M_G$ matrices $ \u_{R,g} \u_{T,g}^H + (\u_{R,g} \u_{T,g}^H)^T $, $g=1,\ldots, G$. 

\section{Extensions}
\subsection{MISO and SIMO channels}
An optimal closed-form solution for the Max-SNR BD-RIS optimization problem can also be obtained when either the transmitter or the receiver is equipped with multiple antennas. Let us take the multiple-input single-output (MISO) case as an example. The single-input multiple-output (SIMO) case is solved analogously. In the MISO case, $\H_T$ is an $M\times N_T$ MIMO channel (we assume that $N_T \leq M$), and $\h_R$ is an $M\times 1$ MISO channel from the RIS to the single-antenna receiver. We assume that the direct link is blocked so we are interested in solving
\begin{align}
\label{eq:MaxSNRproblemBDRIS_SIMO}
({\cal P}_3): \,\max_{\bm{\Theta}}\,\,& \|\h_R^H \Thetab \H_T \|^2 \\
\text{s.t.}\,\,& \,\,\eqref{eq:unitaryconst} \,\, {\text{and}} \,\, \eqref{eq:symmetryconst}. \nonumber 
\end{align}

Notice that $\h_{eq} = \h_R^H \Thetab \H_T = [h_{eq}(1),\ldots, h_{eq}(N_T)]$ is now the equivalent $1 \times N_T$ MISO channel, whose entries are the equivalent SISO channels $h_{eq}(i) = \h_R^H \Thetab \H_T(:,i)$, $i=1,\ldots N_T$, where $\H_T(:,i)$ is the channel from the $i$th Tx antenna to the RIS. 
 
\begin{corollary}
  The optimal symmetric and unitary BD-RIS matrix that maximizes $\|\h_R^H \Thetab \H_T \|^2$ is $\Thetab= \Q \Q^T$, where $\Q$ is a unitary matrix obtained from Takagi's factorization of
\begin{equation}
 \A=  \u_{R} \u_{T,1}^H + (\u_R \u_{T,1}^H )^T = \Q \Sigmab \Q^T,
\end{equation}
where $\u_{R} = \frac{\h_R}{\|\h_R\|}$ and $\u_{T,1}$ is the largest left singular vector of $\H_T = \U_{T} \Lambdab_T \V_T^H = \sum_{i=1}^{N_T} \lambda_{T,i} \u_{T,i}\v_{T,i}^H$.
\end{corollary}

 The proof follows from the fact that, as pointed out in \cite{Nerino2021Arxiv}, the maximum of $\|\h_R^H \Thetab \H_T \|^2$ is achieved  when $|\u_R^H \Thetab \u_{T,1} |^2$ is maximized. Then we simply have to apply the result of Proposition \ref{lemmaTakagi} to a SISO system with channels $\u_R$ and $\u_{T,1}$. 
\begin{remark}
    Once the optimal BD-RIS has been obtained, the equivalent channel $\h_{eq} = \h_R^H \Thetab \H_T$ is known and the transmitter can apply, for example, the maximum ratio transmission (MRT) beamformer ${\bf w}_{T} = \h_{eq}^H/\|\h_{eq}\|$. With this solution, $\|\h_{eq}\|^2$ attains its maximum value $P_{max} = \lambda_{T,1}^2 \|\h_{R}\|^2$, thus solving the Max-SNR joint Tx-beamforming BD-RIS optimization problem.
\end{remark}

\subsection{Sum-rate maximization in the $K$-user SISO MAC}
Let us consider a $K$-user SISO Gaussian MAC in which the direct channels are blocked so that communication is only possible through the RIS. We can form the $M\times K$ matrix $\H_T = [\sqrt{P_1}\h_1, \ldots,\sqrt{P_K}\h_K]$ whose $k$th column represents the channel between the $k$th user and the RIS scaled by the square root of transmitted power. The sum rate is \cite{GoldsmithJSAC}
\begin{align*}
    \sum_{k=1}^K R_k &= \log \left( 1 + \frac{\sum_{k=1}^K P_k |\h_R^H \Thetab \h_k |^2} {\sigma^2} \right) \nonumber\\
    &= \log \left( 1 + \frac{ \|\h_R^H \Thetab \H_T \|^2 } {\sigma^2} \right). \nonumber
\end{align*}
Therefore, finding the BD-RIS that maximizes the sum rate in the $K$-user SISO MAC is equivalent to problem ${\cal P}_3$ in \eqref{eq:MaxSNRproblemBDRIS_SIMO}. 

\section{Simulation Results}
In this section, we evaluate the performance of the Max-SNR BD-RIS solution in two different scenarios: a SISO link and a multiple access channel (MAC). The coordinates $(x,y,z)$ in meters of the
Rx are $(50,0,2)$, the RIS is located at $(40,0,5)$, and the Tx (or the single-antenna users in the MAC) is (are) randomly deployed in a circle in the $(x,y)$ plane with a radius of 10 m centered at (0,0) and at a height $z=2$. We average the result of 100 Monte Carlo simulations. The large-scale path loss in dB is given by $PL = -30 - 10 \beta \log_{10} d$, where $d$ is the link distance in meters, and $\beta=3,75$ is the path-loss exponent. The direct links are blocked and the Tx-RIS and RIS-Rx links are assumed to be uncorrelated Rayleigh channels. For additional details, we refer the reader to \cite{SoleymaniTSP2023,soleymani2022improper}.

\subsection{SNR gain in SISO systems with BD-RIS}
In the first experiment, we evaluate the SNR gain of the BD-RIS using the proposed Max-SNR design for different group sizes with respect to the diagonal RIS that uses optimal phases $\theta_m = -\arg(\h_T(m)\h_R(m)^*)$, $\forall m$. The SNR gain is calculated as
\begin{equation*}
    {\text {\rm SNR Gain}} = 20 \log_{10}  \frac{|\h_R^H \Thetab_{BD-RIS} \h_T |}{|\h_R^H \Thetab_{RIS} \h_T |}.
\end{equation*}
The results are shown in Fig. \ref{fig:FigSPLetters_Rayleigh} for the fully-connected BD-RIS and for the group-connected BD-RIS with group sizes $M_G=$ 2, 4, and 8. The gain in the fully-connected case increases rapidly with the number of RIS elements approaching a value close to 2 dB of gain. Using a group size of $M_G=4$, a gain of about 1.4 dB is achieved, thus providing this value of the group size a good compromise between performance and complexity. For Rician fading channels, the gain between the BD-RIS and the diagonal RIS decreases with increasing Rician K-factor. For pure line-of-sight (LoS) channels, all RIS architectures provide the same SNR. These results are in agreement with the theoretical analysis in \cite{ClerckxTWC23}, where it is shown that for Rayleigh channels when $M \to \infty$, the power gain of a fully-connected BD-RIS over a diagonal RIS converges to $10\log_{10}(16/\pi^2) \approx 2.1$ dB (cf. Eq. (62) in \cite{ClerckxTWC23}). 
\subsection{Maximizing the sum rate in the SISO MAC}
In the second example, we evaluate the sum rate for a 2-user SISO MAC assisted by a BD-RIS. We consider a noise power
spectrum density of $-174 $ dBm/Hz, a system bandwidth of 2 MHz, and a transmit power of 
P = 20 dBm for the 2 users in the uplink. Fig. \ref{fig:SumRate} shows the sum rate for the optimal BD-RIS vs. the number of RIS elements $M$. We include in the plot the sum rate achieved by a conventional (diagonal) RIS with optimized phases \cite{zeng2020sum, Wang18, WangTVT2020, HuangTVT2020} and with random phases. 
\begin{figure}[ht]
    \centering
\includegraphics[width=.5\textwidth]{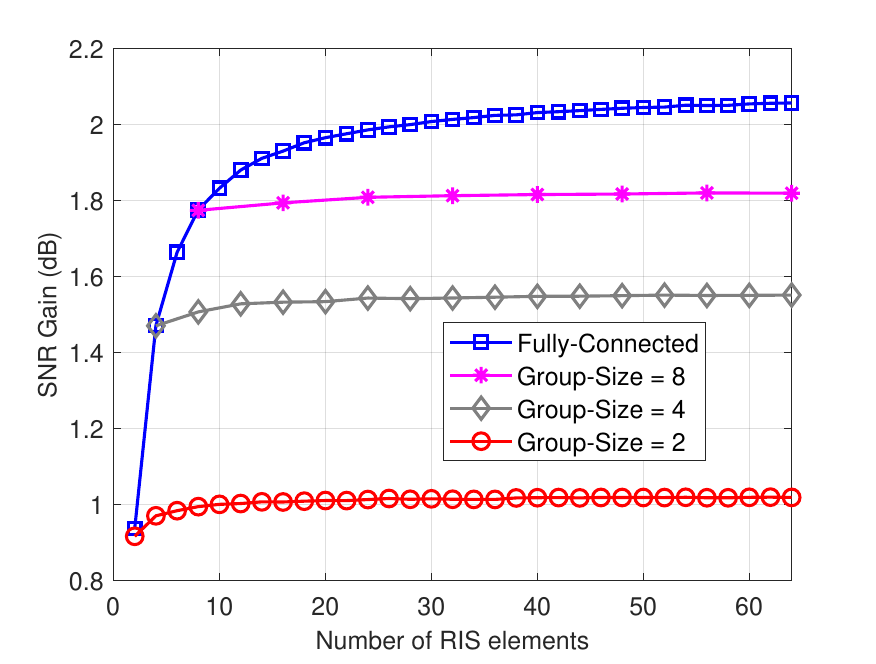}
     \caption{SNR gain in dBs of a BD-RIS compared to a diagonal RIS in a SISO channel with blocked direct link.}
	\label{fig:FigSPLetters_Rayleigh}
\end{figure}

\begin{figure}[ht]
    \centering
\includegraphics[width=.5\textwidth]{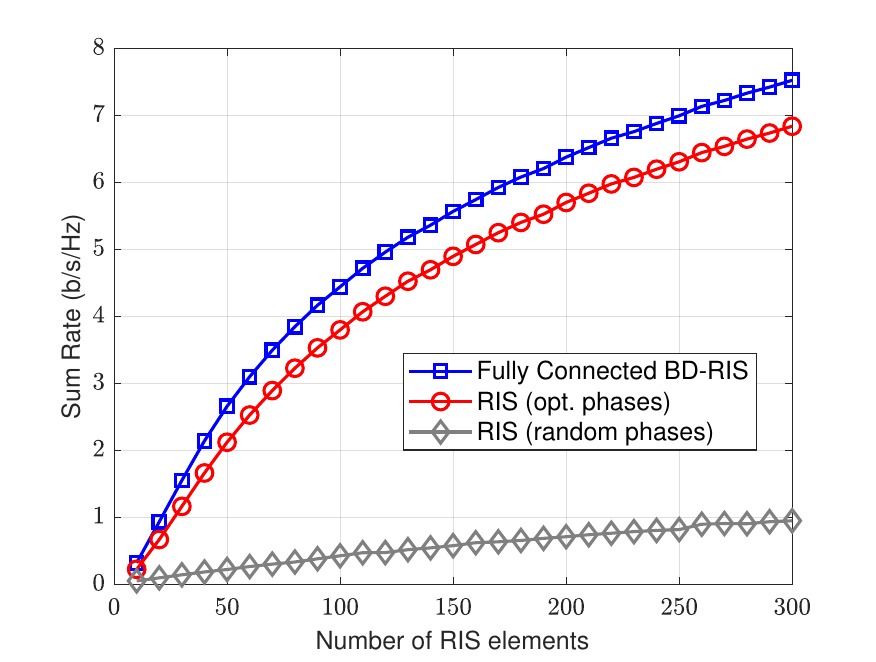}
     \caption{Sum rate vs. $M$ for a 2-user SISO MAC with a fully connected BD-RIS, a diagonal RIS with optimized phases, and a diagonal RIS with random phases.}
	\label{fig:SumRate}
\end{figure}

\section{Conclusion}
The Takagi factorization of a certain complex symmetric matrix allows us to solve the problem of maximizing the SNR in SISO and multi-antenna links assisted by a beyond-diagonal RIS. Some open problems along this line are extensions to the multiple-input multiple-output (MIMO) case, as well as to muti-antenna scenarios (SIMO, MISO, and MIMO) in which the direct channels are not blocked. It would also be interesting to consider the optimization of other metrics in networks assisted by a BD-RIS, such as the capacity in point-to-point MIMO links, the weighted sum rate in the MAC, or the interference leakage in interference channels \cite{Santamaria23}.

\balance
\newpage
\bibliographystyle{IEEEtran}
\bibliography{biblio}
\end{document}